\documentclass[journal]{IEEEtran}
\usepackage{amsmath}
\usepackage{amsfonts}
\usepackage{amsthm}
\usepackage{color}
\usepackage{soul}
\usepackage{graphicx}

\newtheorem{theorem}{Theorem}
\newtheorem{lemma}{Lemma}
\allowdisplaybreaks

\begin{document}
\title{Asymptotically Optimal Massey-Like Inequality on Guessing Entropy With Application to Side-Channel Attack Evaluations}

\author{Andrei T\u{a}n\u{a}sescu, Marios O. Choudary, Olivier Rioul, and Pantelimon George Popescu
\thanks{A. T\u{a}n\u{a}sescu, M.O. Choudary, and P.G. Popescu are with the Department of Computer Science and Engineering, University POLITEHNICA of Bucharest, Splaiul Independe\c{t}ei 313, (6), Bucharest, Romania.}
\thanks{O. Rioul is with the LTCI, T\'{e}l\'{e}com Paris, Institut Polytechnique de Paris, 91120, Palaiseau, France.}
\thanks{Correspondence: pgpopescu@yahoo.com}%
}

\markboth{
}%
{T\u{a}n\u{a}sescu \MakeLowercase{\textit{et al.}}: Asymptotically Optimal Massey-like Inequality on Guessing Entropy With Application to Side-Channel Attack Evaluations}
\maketitle

\begin{abstract}
A Massey-like inequality is any useful lower bound on guessing entropy in terms of the computationally scalable Shannon entropy.
The asymptotically optimal Massey-like inequality is determined and further refined for finite-support distributions. The impact of these results are highlighted for side-channel attack evaluation where guessing entropy is a key metric. In this context, the obtained bounds are compared to the state of the art.
\end{abstract}

\IEEEpeerreviewmaketitle
\section{Introduction}
\IEEEPARstart{T}{he} guessing entropy associated to a (positive descending) probability distribution $\mathbf{p}=\left(p_1,\,p_2,\,\dots,\,p_n\right)$ with ${p}_1\geq \dots \geq {p}_n>0$ is the expected value of the random variable $G\left(\mathbf{p}\right)$ given by $\mathbb{P}\left[G\left(\mathbf{p}\right)=i\right]={p}_i$  ($i=1,\ldots,n$), i.e., $\mathbb{E}\left[G\left(\mathbf{p}\right)\right]=\sum_{i=1}^n i {p}_i$. It corresponds to the minimal average number of binary questions required to guess the value of a random variable distributed according to $\mathbf{p}$~\cite{massey1994guessing}. J. Massey has provided a well-known relation between guessing entropy and the Shannon entropy $H\left(\mathbf{p}\right)=-\sum_{i=1}^n {p}_i\log {p}_i$ which reads~\cite{massey1994guessing}  $\mathbb{E}\left[G\left(\mathbf{p}\right)\right]\geq 2^{H\left(\mathbf{p}\right)-2}+1$ when $H\left(\mathbf{p}\right)\geq 2$ bits. 

Massey's inequality has been recently improved in various ways,
yet all known refinements share the same shape. For instance, in an ISIT paper, Popescu and Choudary~{\cite{popescu2019refinement}} proved
\begin{align*}
    \mathbb{E}\left[G\left(\mathbf{p}\right)\right]
    \geq & 2^{H\left(\mathbf{p}\right)+2{p}_n-2}+1-{p}_n\\
    \geq & 2^{H\left(\mathbf{p}\right)+{p}_n-2}+1-\frac{1}{2}{p}_n\\
    \geq & 2^{H\left(\mathbf{p}\right)-2}+1,
\end{align*}
subject to the same condition $H\left(\mathbf{p}\right)\geq 2$ bits as in the Massey inequality. Meanwhile, Rioul's inequality~\cite{note}, published in a CHES paper~{\cite{de2019best}} states that for all values of $H\left(\mathbf{p}\right)\geq 0$,
\begin{equation}
    \mathbb{E}\left[G\left(\mathbf{p}\right)\right] > \frac1e 2^{H\left(\mathbf{p}\right)},\label{eq:rioul}
\end{equation}
which refines Massey's inequality when $H\left(\mathbf{p}\right)\geq \log\frac{e}{1-e/4}$. Finally, in an Entropy paper, Tanasescu and Popescu~\cite{tuanuasescu2020exploiting} found that under the same condition as in Massey's inequality,
\begin{align*}
    \mathbb{E}\left[G\left(\mathbf{p}\right)\right]\geq&\sup_{\alpha\in\left[0,1/2\right]} 2^{H\left(\mathbf{p}\right)+\frac{h\left(\alpha\right)}{1-\alpha}{p}_n-2}+1 -\frac{\alpha}{1-\alpha}{p}_n\label{eq:gap}
    \\
    \geq &2^{H\left(\mathbf{p}\right)+2{p}_n-2}+1 -{p}_n> 2^{H\left(\mathbf{p}\right)-2}+1\nonumber.
\end{align*}
The authors of~\cite{tuanuasescu2020exploiting} hinted that a similar refinement can be found for inequality~\eqref{eq:rioul}.

In this paper, we optimize exponential relations between the guessing and Shannon entropies, i.e., lower bounds of the form $\mathbb{E}\left[G\left(\mathbf{p}\right)\right]\geq a\cdot b^{H\left(\mathbf{p}\right)} + c$ valid when the Shannon entropy lies above a given threshold. We arrive at an improved Rioul's inequality~\cite{rioul2021variations} by an additive constant of $1/2$, which is asymptotically optimal among other global lower bounds depending only on the Shannon entropy as $H\left(\mathbf{p}\right)\rightarrow \infty$. Then, using the techniques of~\cite{popescu2019refinement,tuanuasescu2020exploiting} we further refine this inequality for finite support distributions allowing us to increase the multiplicative constant depending on the smallest probability $p_n$. Finally, we apply our results to side-channel attack evaluation, where guessing entropy is a key metric~\cite{mazumdar2013constrained,choudary2017efficient,carre2020persistent}, comparing our results to the best on the market and showing that under certain conditions the Shannon entropy is indeed a precious quantifier of guessing entropy.

\section{The Asymptotically Optimal Massey-Like Inequality}
In this section we consider bounds of the form $\mathbb{E}\left[G\left(\mathbf{p}\right)\right]\geq a\cdot b^{H\left(\mathbf{p}\right)} + c$ with $a > 0$ and seek to determine the optimal coefficients $a,\,b,\,c$ prioritizing the asymptotic shape as $H\left(\mathbf{p}\right)\rightarrow \infty$ holding whenever $H\left(\mathbf{p}\right)$ is larger then a given threshold. 

\begin{theorem}
The optimal Massey-like inequality $\mathbb{E}\left[G\left(\mathbf{p}\right)\right]\geq a\cdot b^{H\left(\mathbf{p}\right)}+c$ as $H\left(\mathbf{p}\right)\rightarrow \infty$ is Rioul's improved inequality~\cite{rioul2021variations}
\begin{equation}
    \mathbb{E}\left[G\left(\mathbf{p}\right)\right]\geq \frac{1}{e}2^{H\left(\mathbf{p}\right)}+\frac{1}{2},\label{ineq:rioulimproved}
\end{equation}
which holds for all values of $H\left(\mathbf{p}\right)\geq 0$.
\end{theorem}
\begin{proof}
Following Massey's approach~\cite{massey1994guessing}, finding the best lower bound on guessing entropy is equivalent with the statement that among all probability distributions with guessing entropy $\mu>1$, the maximal Shannon entropy is attained by the geometric distribution with mean $\mu$, that is,
\begin{equation*}
    H\left(\mathbf{p}\right)\leq 
    \log\left(\mu-1\right)-\mu\log\left(1-1/\mu\right)
\end{equation*}
where $\log()$ denotes logarithm to base 2.
The inequality is actually strict when $\mathbf{p}$ has finite length, but the upper bound can be approached as closely as desired.

We seek bounds of the form $\mathbb{E}\left[G\left(\mathbf{p}\right)\right]\geq a\cdot b^{H\left(\mathbf{p}\right)}+c$, i.e. $H\left(\mathbf{p}\right)\leq \log_b\frac{\mu-c}{a}$. In order for this to be valid for all $\mu$, we should necessarily have
\begin{equation*}
    \log_b\frac{\mu-c}{a}\geq\log\left(\mu-1\right)-\mu\log\left(1-1/\mu\right).
\end{equation*}
In particular, as $\mu\rightarrow \infty$, the expression on the left has asymptotic
\begin{equation*}
    \log_b\frac{\mu-c}{a}=\log_b \mu-\log_b{a}-\frac{c\log_b e}{\mu}+o\left(1/\mu\right),
\end{equation*}
while the expression on the right has asymptotic
\begin{equation*}
\log\left(\mu-1\right)-\mu\log\left(1-1/\mu\right)=\log\mu+\log e-\frac{\log e}{2\mu}+o\left(1/\mu\right).
\end{equation*}
As a consequence we necessarily have $\log_b\mu\geq \log\mu$, i.e. $\log b\leq 1$ or $b\leq 2$, so that the optimal (maximum) value of $b$ is $b=2$. Next, we should have $-\log a\geq \log e$, i.e. $a\leq 1/e$, so that the optimal (maximum) value of $a$ is $1/e$. Finally, we should have $-c\log e\geq -\left(\log e\right)/2$, i.e. $c\leq 1/2$, so that the optimal (maximum) value of $c$ is $c=1/2$.

The asymptotically optimal bound then writes
\begin{equation}
    \log\left(\mu-1/2\right)+\log e\geq\log\left(\mu-1\right)-\mu\log\left(1-1/\mu\right)
    \label{asymptopt}
\end{equation}
which readily gives~\eqref{ineq:rioulimproved} when $\mu$ or $H(\mathbf{p})$ tend to infinity. 
A simple proof of~\eqref{ineq:rioulimproved} for all values of $H(\mathbf{p})>0$ can be found in~\cite{rioul2021variations}, but one can also prove directly that~\eqref{asymptopt}  holds for all values of $\mu>1$ as follows. 
The first and second-order derivatives of the difference $f(\mu)=\ln\left(\mu-1/2\right)+1-\ln\left(\mu-1\right)+\mu\ln\left(1-1/\mu\right)$ between the two sides of~\eqref{asymptopt} (expressed in natural units) are 
\begin{align*}
    f'(\mu)
    &=\frac{1}{\mu-1/2}+\ln\bigl(1-\frac1\mu\bigr)\\
f''(\mu)&=-\frac1{(\mu-1/2)^2}+\frac1{\mu(\mu-1)}
 = \frac1{4\mu(\mu-1)(\mu-1/2)^2}.
\end{align*}
It follows that $f''>0$, so that $f'$ is increasing while also vanishing as $\mu\to+\infty$, hence $f'<0$ for all $\mu>1$. As a consequence, $f$ is decreasing for all $\mu>1$. Therefore, since~\eqref{asymptopt} holds when $\mu\to+\infty$, it also holds for all $\mu>1$.  
%
\end{proof}

We conclude this section by remarking that the obtained optimal inequality~\eqref{ineq:rioulimproved} only improves~\eqref{eq:rioul} by an additive constant $1/2$. Rioul's strengthened inequality~\cite{rioul2021variations} now writes
\begin{equation}
    \mathbb{E}\left[G\left(\mathbf{p}\right)\right]\geq \frac{1}{e}2^{H\left(\mathbf{p}\right)}+\frac{1}{2}.\tag{\ref{ineq:rioulimproved}}
\end{equation}
It is further generalized to scalable R\'{e}nyi entropies in~\cite{rioul2021variations}.

\section{Refinement for Finite Support Distributions}
In this section we find a new relation between the Shannon and guessing entropy, dependent on the minimal probability of a given distribution, further refining Rioul's improved inequality~\eqref{ineq:rioulimproved}.

We begin with a direct improvement following the technique of~\cite{popescu2019refinement,tuanuasescu2020exploiting}.
\begin{lemma}
{For any positive descending probability distribution $\mathbf{p}\in\mathbb{R}^n$ such that $H\left(\mathbf{p}\right)\geq 1$ bit, we have}\label{rem:Rioul}
\begin{align*}
    \mathbb{E}\left[G\left(\mathbf{p}\right)\right]
    \geq & \sup_{\alpha\in\left[0,1/2\right]} \frac{1}{e}2^{H\left(\mathbf{p}\right)+{p}_nh\left(\alpha\right)}-\alpha{p}_{n}+\frac{1}{2}\\
    \geq & \frac{1}{e}2^{H\left(\mathbf{p}\right)+{p}_n}-\frac{1}{2}{p}_{n}+\frac{1}{2}\geq \frac{1}{e}2^{H\left(\mathbf{p}\right)}+\frac{1}{2}.
\end{align*}
\end{lemma}
\begin{proof}
Consider a positive decreasing distribution $\mathbf{p}=\left(p_1,\,p_2,\,\dots,\,p_n\right)$ with $H\left(\mathbf{p}\right)\geq 2$. Following the approach in~\cite{popescu2019refinement} we construct the new probability distribution $\mathbf{q}=\left(p_1,\,p_2,\,\dots,\,p_{n-1},\,\left(1-\alpha\right){p}_n,\,\alpha{p}_n\right)$, which is decreasing and strictly positive if and only if $\alpha\in\left(0,\,1/2\right]$. From the grouping property of entropy, $H\left(\mathbf{q}\right)=H\left(\mathbf{p}\right)+{p}_n h\left(\alpha\right)$, and moreover $\mathbb{E}\left[G\left(\mathbf{q}\right)\right]=\mathbb{E}\left[G\left(\mathbf{p}\right)\right]+\alpha{p}_{n}$. Then
\begin{align}
    \mathbb{E}\left[G\left(\mathbf{p}\right)\right] =& \mathbb{E}\left[G\left(\mathbf{q}\right)\right] - \alpha p_n 
    \geq \frac{1}{e}2^{H\left(\mathbf{q}\right)}-\alpha{p}_{n}+\frac{1}{2}\label{eq:halp_rioul}\\ =& \frac{1}{e}2^{H\left(\mathbf{p}\right)+{p}_nh\left(\alpha\right)}-\alpha{p}_{n}+\frac{1}{2}.\nonumber
\end{align}
{The first inequality follows taking the {supremum} over $\alpha$ in eq.~{\eqref{eq:halp_rioul}}, the second by substituting $\alpha=1/2$. To justify the third, we use $2^x>1+x\ln2$ for $x=p_n$ obtaining}
\begin{align*}
    \frac{1}{e}2^{H\left(\mathbf{p}\right)+{p}_n}-\frac{1}{2}{p}_{n} \geq& \frac{1}{e}2^{H\left(\mathbf{p}\right)} \left( 1+ {p}_n\ln\,2\right) -\frac{1}{2}{p}_{n}\\ =& \frac{1}{e}2^{H\left(\mathbf{p}\right)} + \left(\frac{2^{H\left(\mathbf{p}\right)}\ln 2}{e}-\frac{1}{2}\right) p_n,
\end{align*}
{where $p_n$'s coefficient is positive whenever $H\left(\mathbf{p}\right)\geq \log \frac{e}{2\ln2}$. This ends the proof.}
\end{proof}

We can further refine this lemma using the techniques of~\cite{popescu2019refinement,tuanuasescu2020exploiting} as follows.
\begin{theorem}
For any positive descending probability distributions $\mathbf{p}\in\mathbb{R}^n$ such that $H\left(\mathbf{p}\right)\geq 1$, we have\label{thm:RioulGen}
\begin{align*}
    \mathbb{E}\left[G\left(\mathbf{p}\right)\right]
    \geq & \sup_{\alpha\in\left[0,1/2\right]} \frac{1}{e}2^{H\left(\mathbf{p}\right)+\frac{h\left(\alpha\right)}{1-\alpha}{p}_n}+\frac{1}{2}-\frac{\alpha}{1-\alpha}{p}_{n}\\
    \geq&  \frac{1}{e}2^{H\left(\mathbf{p}\right)+\frac{1}{2}{p}_n}+\frac{1}{2}-{p}_{n}\geq \frac{1}{e}2^{H\left(\mathbf{p}\right)}+\frac{1}{2}.
\end{align*}
\end{theorem}
\begin{proof}
Given the initial decreasing $\mathbf{p}$, we construct a sequence of probability distributions $\left\{\mathbf{Q}_k\right\}$, recursively defined using the procedure in the previous proof.

We begin by fixing an arbitrary parameter $\alpha\in\left[0,1/2\right]$ as above. 
Denoting by $Q_{k,i}$ the $i$\textsuperscript{th} component of the sequence $\mathbf{Q}_k$, we define the terms of the list $\left\{\mathbf{Q}_k\right\}$ as follows. We let the support of the first term coincide with $\mathbf{p}$, i.e. $\mathbf{Q}_{0}=\left(p_0,\,p_1,\,\dots,\,p_n,\,0,\,0,\,\dots,\,0,\,\dots\right)$, and we define the other terms by recurrence: 
\begin{align*}
    \mathbf{Q}_{k+1}=&\left(Q_{k,0},\,Q_{k,1},\,\dots,\,Q_{k,n+k-1},\,\right.\\
    &\left.\left(1-\alpha\right)Q_{k,n+k},\,\alpha Q_{k,n+k},\,0,\,0,\,\dots,\,0,\,\dots\right).
\end{align*}
and at each step of the construction we have the inequality
\begin{align*}
    \mathbb{E}\left[G\left({\mathbf{Q}_k}\right)\right]
    =&\mathbb{E}\left[G\left({\mathbf{Q}_{k+1}}\right)\right] -\alpha{Q}_{k,n+k}\\
    \geq& \frac{2^{H\left(\mathbf{Q}_{k+1}\right)}}{e}-\alpha{Q}_{k,n+k}
    > \frac{2^{H\left(\mathbf{\mathbf{Q}_k}\right)}}{e}+\frac{1}{2}.
\end{align*}
After the first $k$ steps of the construction we find
{\small{\begin{align*}
    \mathbb{E}\left[G\left(\mathbf{p}\right)\right]
    =&\mathbb{E}\left[G\left({\mathbf{Q}_{k}}\right)\right]-{p}_n\alpha\frac{1-\alpha^k}{1-\alpha}\\
    =&\mathbb{E}\left[G\left({\mathbf{Q}_{k}}\right)\right] + \sum_{j=0}^{k-1} \left(\mathbb{E}\left[G\left({\mathbf{Q}_{j}}\right)\right] - \mathbb{E}\left[G\left({\mathbf{Q}_{j+1}}\right)\right]\right)\\
    \geq&\frac{1}{2}2^{H\left(\mathbf{Q}_{k}\right)}+\frac{1}{2} + \sum_{j=0}^{k-1} \left(\mathbb{E}\left[G\left({\mathbf{Q}_{j}}\right)\right] - \mathbb{E}\left[G\left({\mathbf{Q}_{j+1}}\right)\right]\right)\\
    >&\frac{1}{e}2^{H\left(\mathbf{Q}_{k-1}\right)}+\frac{1}{2} + \sum_{j=0}^{k-2} \left(\mathbb{E}\left[G\left({\mathbf{Q}_{j}}\right)\right] - \mathbb{E}\left[G\left({\mathbf{Q}_{j+1}}\right)\right]\right)\\
    >&\dots > \frac{1}{e}2^{H\left(\mathbf{Q}_0\right)}+\frac{1}{2}=\frac{1}{e}2^{H\left(\mathbf{p}\right)}+\frac{1}{2},
\end{align*}}}
where the tightest of the enumerated bounds is
{\small{\begin{align*}
    \mathbb{E}\left[G\left(\mathbf{p}\right)\right]\geq& \frac{1}{e}2^{H\left(\mathbf{Q}_{k}\right)}+\frac{1}{2} + \sum_{j=0}^{k-1} \left(\mathbb{E}\left[G\left({\mathbf{Q}_{j}}\right)\right] - \mathbb{E}\left[G\left({\mathbf{Q}_{j+1}}\right)\right]\right)\\
    =&\frac{1}{e}2^{H\left(\mathbf{p}\right)+{p}_nh\left(\alpha\right)\frac{1-\alpha^{k}}{1-\alpha}}+\frac{1}{2} -{p}_n\alpha\frac{1-\alpha^k}{1-\alpha},
\end{align*}}}
which as we have shown increases with $k$ up to the limit
\begin{align*}
    \mathbb{E}\left[G\left(\mathbf{p}\right)\right]\geq \frac{1}{e}2^{H\left(\mathbf{p}\right)+{p}_n\frac{h\left(\alpha\right)}{1-\alpha}}+\frac{1}{2} -{p}_n\frac{\alpha}{1-\alpha}
\end{align*}
valid for any $\alpha\in\left[0,1/2\right]$.
The first desired inequality now follows taking \emph{supremum} over the last equation, the second by substituting $\alpha=1/2$ and the third by noting that all bounds in the sequence are greater than the last one $\frac{1}{e}2^{H\left(\mathbf{p}\right)}+\frac{1}{2}$.
\end{proof}

\section{Application to Side-Channel Analysis}

The improvements shown in previous sections can be very useful in the evaluation
of side-channel attacks. In this context, Choudary and Popescu~\cite{ches17} presented a new approach, based on mathematical bounds of the guessing entropy~\cite{massey1994guessing},
to bound the guessing entropy remaining after a side-channel attack for very large
cryptographic keys (or other secret data). They showed that their method works
for keys of up to 1024 bytes and beyond, working in constant time and memory,
which none of the other methods could do. This provided a great improvement for security evaluations of cryptographic devices.

We remark here that all bounds from this paper are highly computationally scalable, because they are based on the Shannon entropy, which is additive i.e. $H\left(\otimes_i\mathbf{P}_i\right)=\sum_iH\left(\mathbf{P}_i\right)$ for any probability distributions $\mathbf{P}_1,\,\mathbf{P}_2,\,\dots,\,\mathbf{P}_n$~\cite{cover1999elements}.

\subsection{Evaluation of bounds}

In this context of security evaluations, it is interesting to evaluate the accuracy of different bounds for the guessing entropy in different settings. In this section, we analyse the bounds derived in the preceding sections, along with those presented at CHES 2017~\cite{ches17}, using lists of probabilities obtained from the application of Template Attacks~\cite{chari03} on side-channel traces. 

For easier comparison and future reference, we used the same data as in the CHES 2017 paper: A simulated dataset (MATLAB generated  power consumption from the execution of the AES S-box) and a real dataset (power traces from the execution of AES in the AES hardware engine of an AVR XMEGA microcontroller).

For our analysis we have focused on three interesting cases: 1) application of the bounds on single lists of probabilities -- this is equivalent to attacking a single key byte in side-channel attack evaluations; 2) application of the bounds on the combination of two bytes -- this is interesting to observe the scalability of the bounds; 3) application of the bounds on the combination of all 16 AES bytes -- this represents a complete attack on the full AES key and hence is a representative scenario of a full-fledged security evaluation.

\begin{figure}[htb!]
  \includegraphics[width=1.7in]{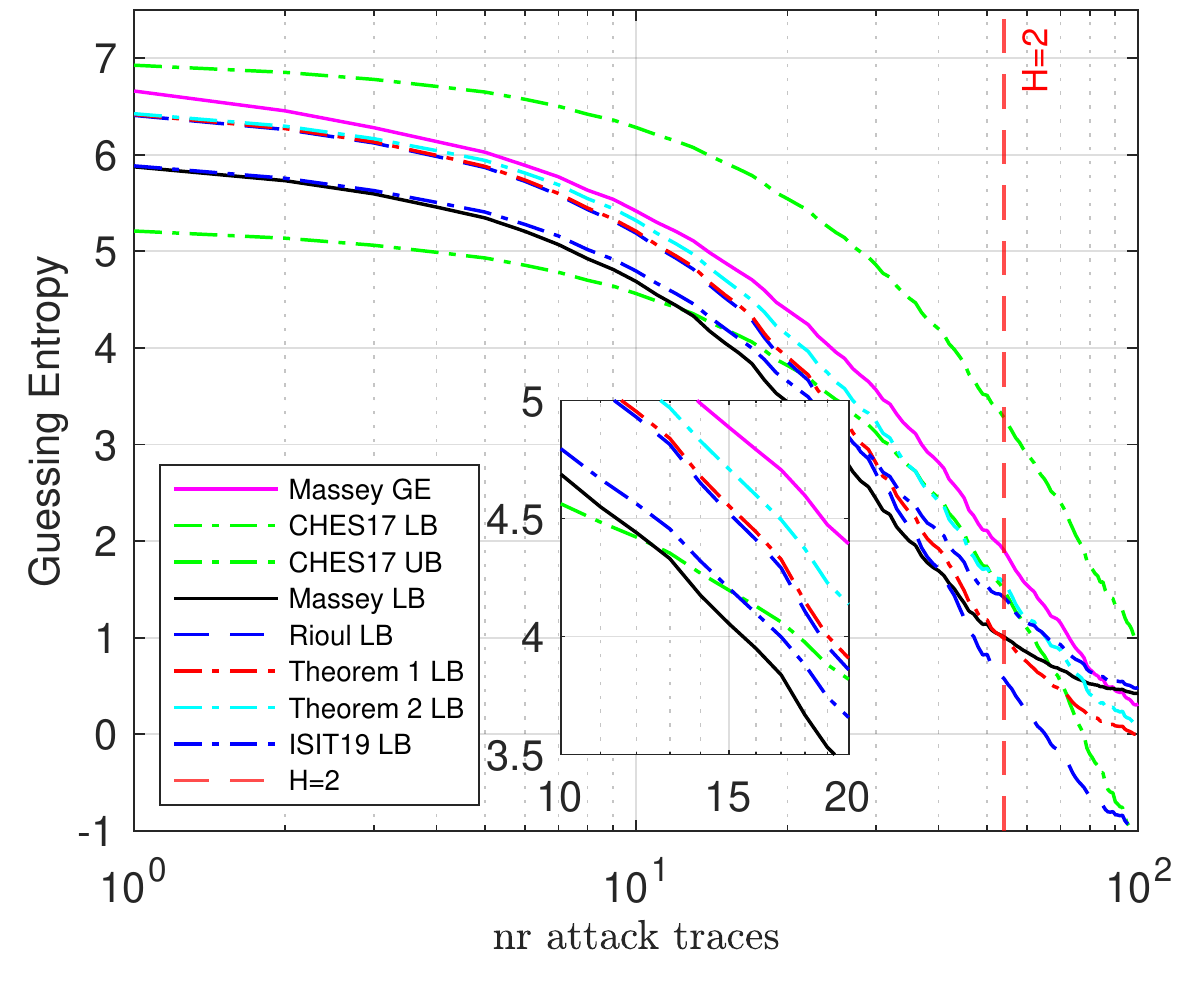}
  \includegraphics[width=1.7in]{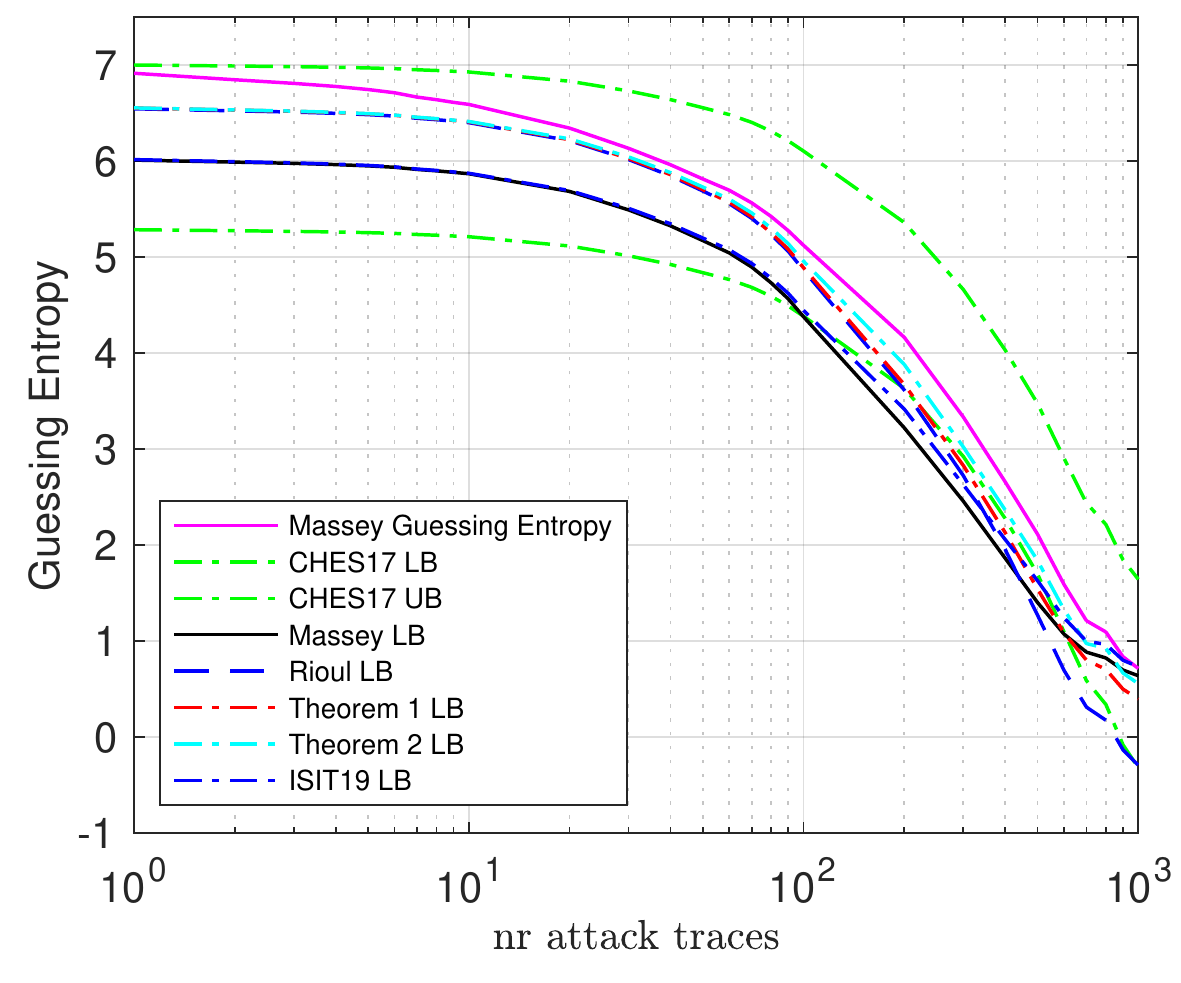}
  \caption{
    \label{fig:bounds1b}
    Bounds for the simulated (left) and real
    (right) datasets, when targeting a single subkey byte. These are averaged
    results over 100 experiments.
    }
\end{figure}

\subsection{Evaluation on a single byte}

We show the bounds for a single key byte on the simulated and real datasets in Figure~\ref{fig:bounds1b}. Here we can see that while the CHES lower bound is tighter when the guessing entropy is low (below 4 bits), in the other (most) cases Rioul's lower bound is better. Furthermore, we can see that Theorem 1 provides a better (tighter) lower bound than Rioul's lower bound and Theorem 2 in turn provides an even better lower bound than Theorem 1.

An interesting artifact appears when the guessing entropy decreases below two bits ($\log(G(\mathbf{p}))=1$), where the Massey inequality (and the ones in ISIT 2019~\cite{popescu2019refinement}) does not necessarily hold (considering for example geometric distributions with $p_1\geq 1/2$). In this case, most bounds seem to be tighter than the CHES 2017~\cite{ches17} lower bound. Meanwhile, bounds based on Rioul's inequality all continue to hold in this regime, owing to the fact that it does not impose preconditions on the minimal value of $H(\mathbf{p})$.

\begin{figure}[htb!]
  \includegraphics[width=1.7in]{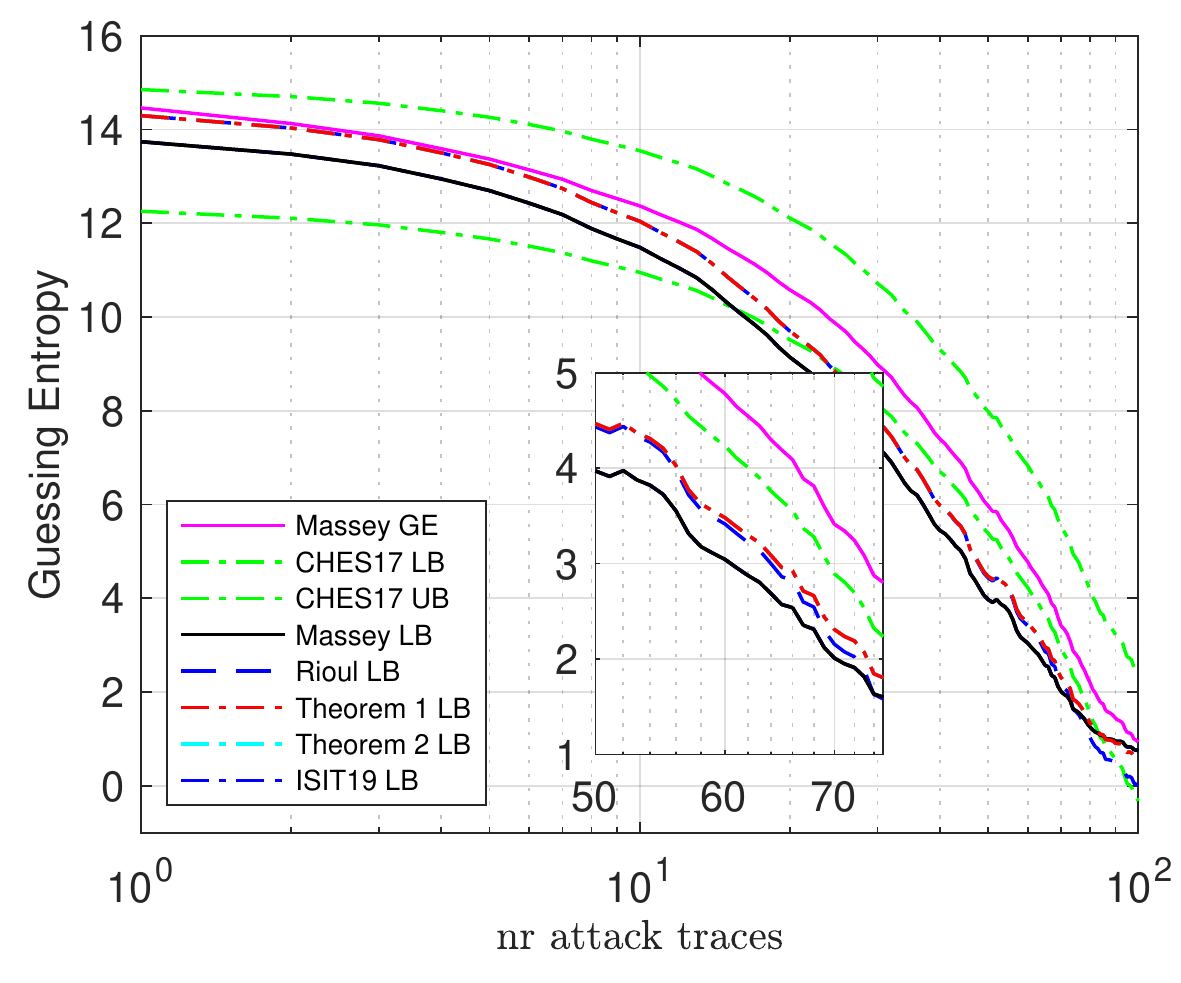}
  \includegraphics[width=1.7in]{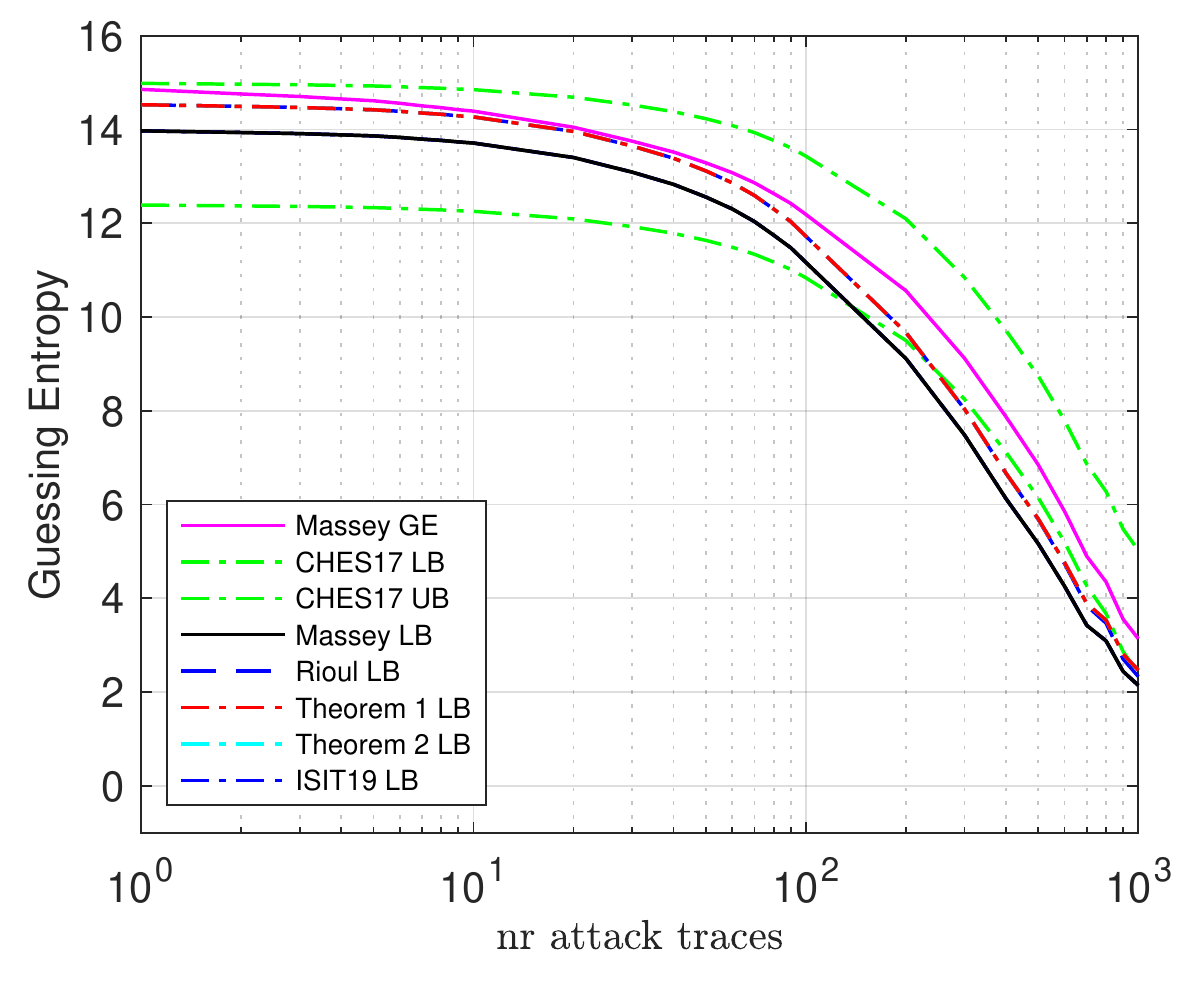}
  \caption{
    \label{fig:bounds2b}
    Bounds for the simulated (left) and real
    (right) datasets, when targeting two subkey bytes. These are averaged
    results over 100 experiments.
    }
\end{figure}

\subsection{Evaluation on two bytes}

We show the bounds when targetting two key bytes on the simulated and real datasets in Figure~\ref{fig:bounds2b}. Here we see again that Rioul's bound is tight when the guessing entropy is higher, but then the CHES lower bound becomes tighter, as the guessing entropy decreases. We can also confirm here that Theorem 1 provides a better (tighter) lower bound than Rioul's lower bound.

However, in this case Theorem 2 provides numerically similar results to Theorem 1, just as the ISIT 2019 lower bound provides numerically similar results to Massey's lower bound. These results are due to the fact that these bounds only differ pairwise in a term containing the minimum probability in the combined list and this minimum becomes zero (or almost zero) when combining two (or more) lists of probabilities in our experiments, which is just a particularity of such experiments.

\begin{figure}[htb]
  \includegraphics[width=1.7in]{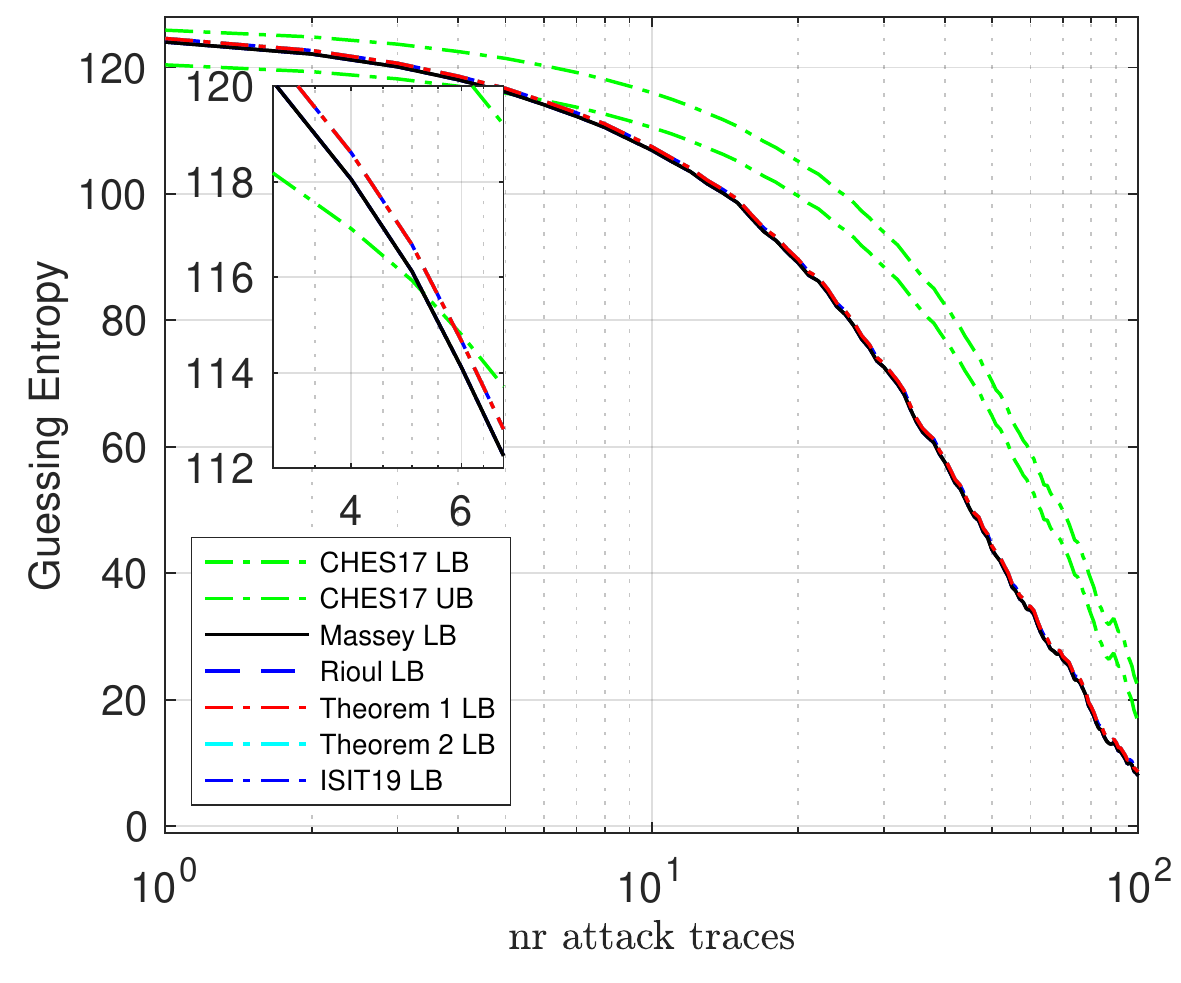}
  \includegraphics[width=1.7in]{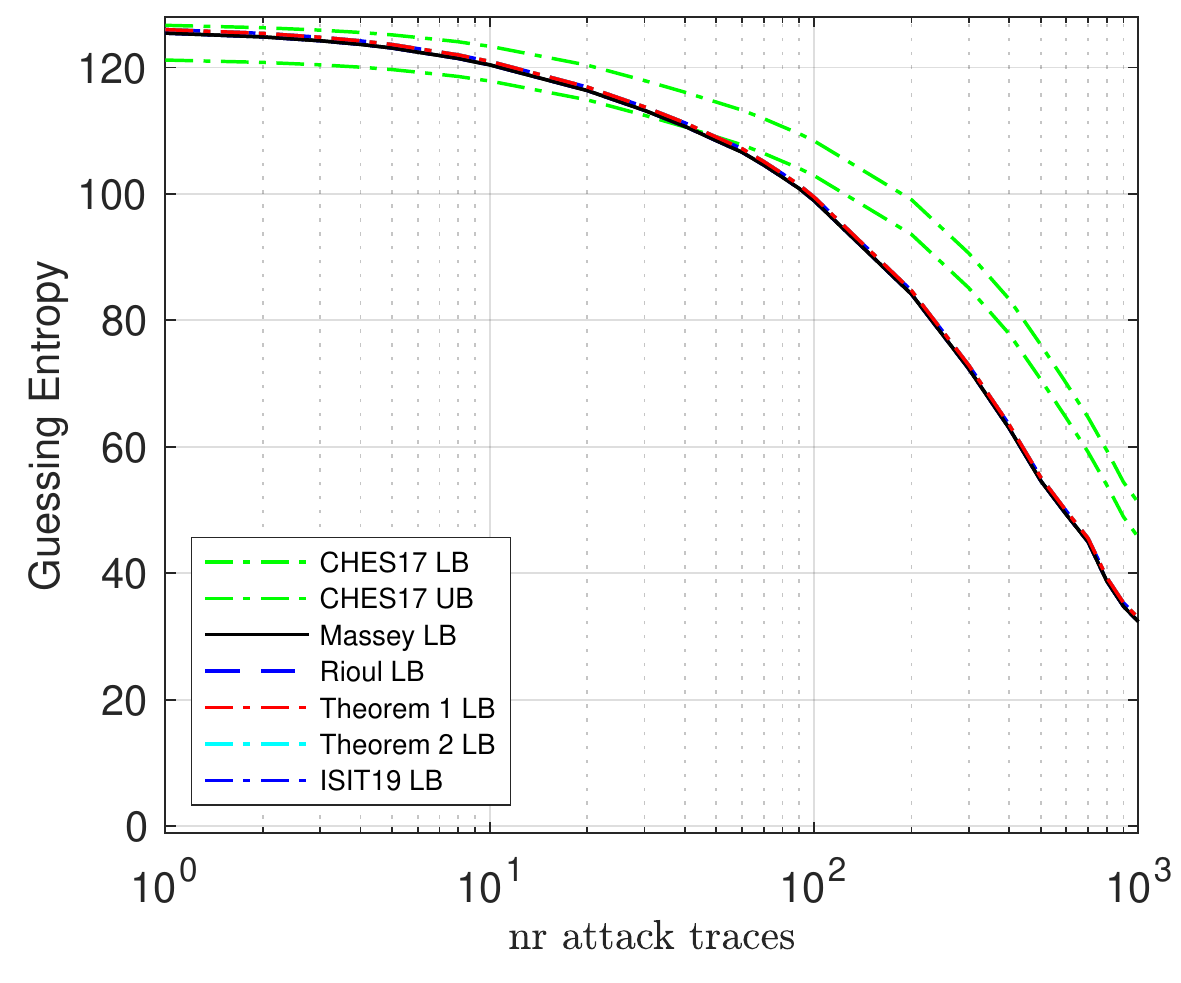}
  \caption{
    \label{fig:bounds16b}
    Bounds for the simulated (left) and real
    (right) datasets, when targeting all the 16 AES key bytes. These are averaged
    results over 100 experiments.
    }
\end{figure}

\subsection{Evaluation on all 16 bytes}

Finally, we show the bounds when targeting all the 16 bytes of the full AES key on the simulated and real datasets in Figure~\ref{fig:bounds16b}.
We did not plot the actual value of the guessing entropy in this case, because it is not possible to compute it: it would require the iteration over (and sorting of)
a list of $2^{128}$ elements. Hence, in this case the computationally efficient bounds compared in this paper become very valuable. From the figure we see again that when the guessing entropy is very high (e.g. above 120 bits), all the lower bounds presented in this paper are tighter than the CHES 2017 lower bound. However, as soon as the guessing entropy decreases below 120 bits, the CHES 2017 lower bound becomes closer to the upper bound than the other lower bounds.

We can also confirm here that Rioul's lower bound is a better (tighter) lower bound than Massey's lower bound. However, in this case we observe that Theorem 1 and 2 provide numerically similar results to Rioul's lower bound. Nevertheless, we are impressed by the scalability of such bounds, thanks to the easy computation of the Shannon entropy of product distributions.

\section{Conclusion}
In this paper, the asymptotically optimal Massey-like inequality is determined as an improved Rioul's inequality by an additive constant of $1/2$. Then, using the techniques of~\cite{popescu2019refinement,tuanuasescu2020exploiting}, this inequality is further refined for finite support distributions allowing us to increase the multiplicative constant depending on the smallest probability $p_n$. Finally, the results are applied to the task of side-channel attack evaluation and compared to the best on the market. It is shown that under certain conditions, the Shannon entropy is in fact a precious quantifier of guessing entropy because  it is computationally scalable thanks to its additivity property. 

For future work we are very interested in further results based on other (additive) entropies, such as R\'{e}nyi entropies where other guessing bounds are already investigated~\cite{rioul2021variations} past their original use in moment inequalities~\cite{arikan1996inequality,sason2018improved,kuzuoka2019conditional} and other derived problems such as guessing with limited (or no) memory~\cite{huleihel2017guessing}.

\section*{Acknowledgment}

This work was partially supported the Romanian Ministry of Education
and Research, CNCS -- UEFISCDI, project number PN-III-P1-1.1-TE-2019-2245, within PNCDI III.


\begin{thebibliography}{10}
\providecommand{\url}[1]{#1}
\csname url@samestyle\endcsname
\providecommand{\newblock}{\relax}
\providecommand{\bibinfo}[2]{#2}
\providecommand{\BIBentrySTDinterwordspacing}{\spaceskip=0pt\relax}
\providecommand{\BIBentryALTinterwordstretchfactor}{4}
\providecommand{\BIBentryALTinterwordspacing}{\spaceskip=\fontdimen2\font plus
\BIBentryALTinterwordstretchfactor\fontdimen3\font minus
  \fontdimen4\font\relax}
\providecommand{\BIBforeignlanguage}[2]{{%
\expandafter\ifx\csname l@#1\endcsname\relax
\typeout{** WARNING: IEEEtran.bst: No hyphenation pattern has been}%
\typeout{** loaded for the language `#1'. Using the pattern for}%
\typeout{** the default language instead.}%
\else
\language=\csname l@#1\endcsname
\fi
#2}}
\providecommand{\BIBdecl}{\relax}
\BIBdecl

\bibitem{massey1994guessing}
J.~L. Massey, ``Guessing and entropy,'' in \emph{Proceedings of 1994 IEEE
  International Symposium on Information Theory}.\hskip 1em plus 0.5em minus
  0.4em\relax IEEE, 1994, p. 204.

\bibitem{popescu2019refinement}
P.~G. Popescu and M.~O. Choudary, ``Refinement of {Massey} inequality,'' in
  \emph{2019 IEEE International Symposium on Information Theory (ISIT)}.\hskip
  1em plus 0.5em minus 0.4em\relax IEEE, 2019, pp. 495--496.

\bibitem{note}
O.~Rioul, ``On guessing,'' unpublished note, 11 2013.

\bibitem{de2019best}
E.~de~Ch{\'e}risey, S.~Guilley, O.~Rioul, and P.~Piantanida, ``Best information
  is most successful,'' in \emph{2019 International Conference on Cryptographic
  Hardware and Embedded Systems (CHES)}, 2019, pp. 49--79.

\bibitem{tuanuasescu2020exploiting}
A.~T{\u{a}}n{\u{a}}sescu and P.~G. Popescu, ``Exploiting the {Massey} gap,''
  \emph{Entropy}, vol.~22, no.~12, p. 1398, 2020.

\bibitem{rioul2021variations}
O.~Rioul, ``Variations on a theme by {Massey},'' \emph{arXiv preprint
  arXiv:2102.04200 (submitted to IEEE Transactions on Information Theory)},
  2021.

\bibitem{mazumdar2013constrained}
B.~Mazumdar, D.~Mukhopadhyay, and I.~Sengupta, ``Constrained search for a class
  of good bijective {$S$}-boxes with improved dpa resistivity,'' \emph{IEEE
  Transactions on Information Forensics and Security}, vol.~8, no.~12, pp.
  2154--2163, 2013.

\bibitem{choudary2017efficient}
M.~O. Choudary and M.~G. Kuhn, ``Efficient, portable template attacks,''
  \emph{IEEE Transactions on Information Forensics and Security}, vol.~13,
  no.~2, pp. 490--501, 2017.

\bibitem{carre2020persistent}
S.~Carr{\'e}, S.~Guilley, and O.~Rioul, ``Persistent fault analysis with few
  encryptions,'' in \emph{2020 International Workshop on Constructive
  Side-Channel Analysis and Secure Design (COSADE)}, 2020.

\bibitem{ches17}
M.~O. Choudary and P.~G. Popescu, ``Back to {Massey}: {Impressively} fast,
  scalable and tight security evaluation tools,'' in \emph{2017 International
  Conference on Cryptographic Hardware and Embedded Systems (CHES)}, 2017, pp.
  367--386.

\bibitem{cover1999elements}
T.~M. Cover and J.~A. Thomas, \emph{Elements of information theory}.\hskip 1em
  plus 0.5em minus 0.4em\relax John Wiley \& Sons, 1999.

\bibitem{chari03}
S.~Chari, J.~R. Rao, and P.~Rohatgi, ``Template attacks,'' in \emph{2003
  Cryptographic Hardware and Embedded Systems (CHES)}.\hskip 1em plus 0.5em
  minus 0.4em\relax Springer Berlin Heidelberg, 2003, pp. 13--28.

\bibitem{arikan1996inequality}
E.~Arikan, ``An inequality on guessing and its application to sequential
  decoding,'' \emph{IEEE Transactions on Information Theory}, vol.~42, no.~1,
  pp. 99--105, 1996.

\bibitem{sason2018improved}
I.~Sason and S.~Verd{\'u}, ``Improved bounds on lossless source coding and
  guessing moments via {R{\'e}nyi} measures,'' \emph{IEEE Transactions on
  Information Theory}, vol.~64, no.~6, pp. 4323--4346, 2018.

\bibitem{kuzuoka2019conditional}
S.~Kuzuoka, ``On the conditional smooth r{\'e}nyi entropy and its applications
  in guessing and source coding,'' \emph{IEEE Transactions on Information
  Theory}, vol.~66, no.~3, pp. 1674--1690, 2019.

\bibitem{huleihel2017guessing}
W.~Huleihel, S.~Salamatian, and M.~M{\'e}dard, ``Guessing with limited
  memory,'' in \emph{2017 IEEE International Symposium on Information Theory
  (ISIT)}.\hskip 1em plus 0.5em minus 0.4em\relax IEEE, 2017, pp. 2253--2257.

\end{thebibliography}



%

\begin{IEEEbiography}[{\includegraphics[width=1in,height=1.25in,clip,keepaspectratio]{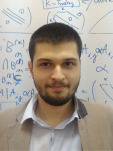}}]{Andrei T\u{a}n\u{a}sescu} is a young researcher at University POLITEHNICA of Bucharest. His main research interests are Quantum Computing and Quantum Information Theory.
\end{IEEEbiography}

\begin{IEEEbiography}[{\includegraphics[width=1in,height=1.25in,clip,keepaspectratio]{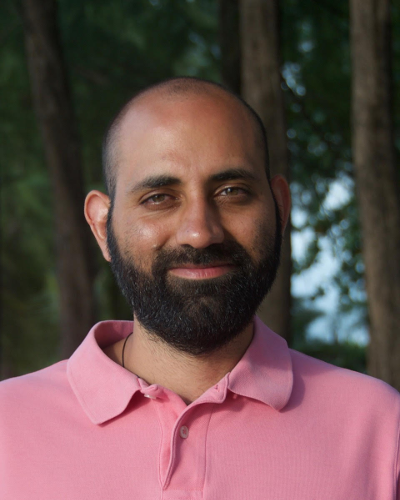}}]{Marios O. Choudary}
is Senior Lecturer in Computer Science at the University Politehnica
of Bucharest. His research interests include authentication and
security protocols, applied cryptography and side-channel attack
evaluations. He graduated from the University Politehnica of
Bucharest in 2008 and then did his MPhil and PhD in Computer Science at the University of Cambridge Computer Laboratory.
\end{IEEEbiography}

\begin{IEEEbiography}[{\includegraphics[width=1in,height=1.25in,clip,keepaspectratio]{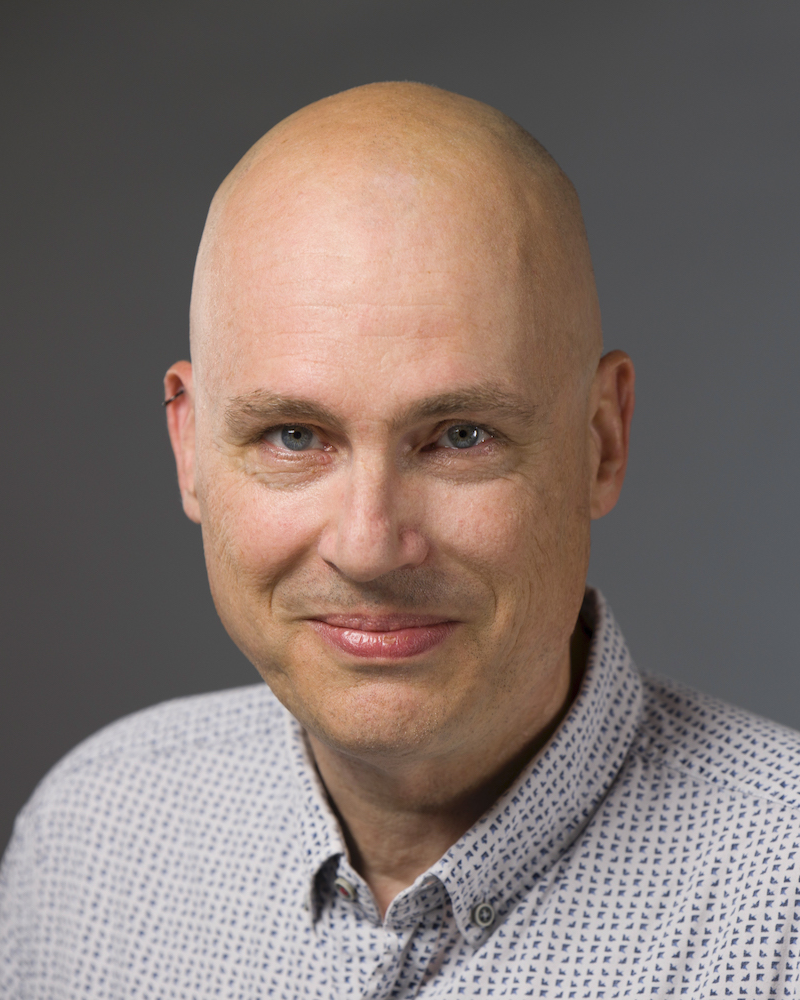}}]{Olivier Rioul} 
is full Professor at the Department of Communication and Electronics, in the Laboratoire de Traitement et Communication de l'Information (LTCI), T\'el\'ecom Paris, Institut Polytechnique de Paris, France. He graduated from École Polytechnique, Paris, France in 1987 and from École Nationale Supérieure des Télécommunications, Paris, France in 1989. He obtained his PhD degree from École Nationale Supérieure des Télécommunications, Paris, France in 1993. His research interests are in applied mathematics and include various, sometimes unconventional, applications of information theory such as inequalities in statistics, hardware security, and experimental psychology. He has been teaching information theory at
various French universities for more than twenty years and has published a textbook which has become a classical French reference in the field. 
\end{IEEEbiography}

\begin{IEEEbiography}[{\includegraphics[width=1in,height=1.25in,clip,keepaspectratio]{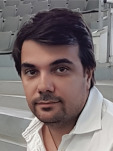}}]{Pantelimon George Popescu} is Professor at the Computer Science and Engineering Department of University POLITEHNICA of Bucharest. His main fields of interest include Quantum Computing, Numerical Methods, Information Theory and Inequalities.
\end{IEEEbiography}





\end{document}